\documentclass{article}

\usepackage{tikz}
\usepackage{pgfplots}
\pgfplotsset{compat=1.17}

\PassOptionsToPackage{numbers, compress}{natbib}

\usepackage{amsthm}

\usepackage{lipsum}

\usepackage[utf8]{inputenc}
\usepackage[T1]{fontenc}

\usepackage{booktabs}
\usepackage{longtable,tabularx}
\usepackage{microtype}
\usepackage{multicol}
\usepackage{multirow}
\usepackage{caption}
\usepackage{savesym}
\usepackage{algorithm}
\usepackage{subfigure}

\usepackage{thm-restate}

\savesymbol{Bbbk}
\usepackage{amsmath,amssymb,amsfonts}
\usepackage{url}
\usepackage{tikz}
\usepackage{mathrsfs}
\usepackage{nicefrac,bbm}
\usepackage{hyperref}
\hypersetup{
  colorlinks=true,
  linkcolor=green!30!black,
  linktoc=all,
  citecolor=blue,
  bookmarksnumbered=true,
  pdfproducer={},
  pdfinfo={CreationDate={2025},ModDate={2025}}, pdfauthor={Sander Borst and Max Springer},
  pdfkeywords={},
  pdfsubject={},
  pdfcreator={},
  pdflang={}
}
\usepackage{cleveref}
\usepackage{mathtools,pifont,enumitem}
\usepackage[framemethod=TikZ]{mdframed}
\usepackage{color,fancybox,graphicx,subfigure}

\savesymbol{st}
\usepackage{soul}
\restoresymbol{SOUL}{st}
\usepackage[normalem]{ulem}

\usepackage{xurl}
\usepackage{scalerel}
\usetikzlibrary{math}
\usepackage{bbding}
\usepackage{array}
\usepackage{dsfont}

\usepackage{empheq}
\usepackage{stackengine}

\newtheorem{theorem}{Theorem}[section]
\newtheorem{lemma}[theorem]{Lemma}
\newtheorem{definition}[theorem]{Definition}

\newtheorem{claim}[theorem]{Claim}

\newtheorem{proposition}[theorem]{Proposition}

\newtheorem{remark}[theorem]{Remark}

\crefname{section}{Section}{Sections}
\crefname{theorem}{Theorem}{Theorems}
\crefname{assumption}{Assumption}{Assumptions}
\crefname{lemma}{Lemma}{Lemmas}
\crefname{definition}{Definition}{Definitions}
\crefname{conjecture}{Conjecture}{Conjectures}
\crefname{corollary}{Corollary}{Corollaries}
\crefname{construction}{Construction}{Constructions}
\crefname{claim}{Claim}{Claims}
\crefname{observation}{Observation}{Observations}
\crefname{proposition}{Proposition}{Propositions}
\crefname{fact}{Fact}{Facts}
\crefname{question}{Question}{Questions}
\crefname{problem}{Problem}{Problems}
\crefname{remark}{Remark}{Remarks}
\crefname{example}{Example}{Examples}
\crefname{equation}{Equation}{Equations}
\crefname{appendix}{Appendix}{Appendices}
\crefname{algorithm}{Algorithm}{Algorithms}
\crefname{model}{Model}{Models}
\crefname{figure}{Figure}{Figures}

\AfterEndEnvironment{definition}{\noindent\ignorespaces}
\AfterEndEnvironment{assumption}{\noindent\ignorespaces}
\AfterEndEnvironment{lemma}{\noindent\ignorespaces}
\AfterEndEnvironment{theorem}{\noindent\ignorespaces}
\AfterEndEnvironment{proposition}{\noindent\ignorespaces}
\AfterEndEnvironment{fact}{\noindent\ignorespaces}
\AfterEndEnvironment{question}{\noindent\ignorespaces}
\usepackage{max-style}

\usepackage{color-edits}

\renewcommand{\epsilon}{\varepsilon}

\usepackage{xfrac}

\newif\iftproofshowtitle
\tproofshowtitlefalse
\newenvironment{tproof}[1][]{\iftproofshowtitle
    \if\relax\detokenize{#1}\relax
      \begin{proof}\else
      \begin{proof}[#1]\fi
  \else
    \begin{proof}\fi
}{\end{proof}}

\newcommand{\agents}{N}
\newcommand{\items}{M}
\newcommand{\class}[1]{N_{#1}}

\newcommand{\Exp}{\mathbb{E}}
\usepackage[most]{tcolorbox}

\tcbset{
  mygraybox/.style={
    colback=gray!20,
    colframe=gray!20,
    boxrule=0pt,
    arc=0pt,
    outer arc=0pt,
    boxsep=0pt,
    left=0pt,
    right=0pt,
    top=0pt,
    bottom=0pt,
    enhanced,
    breakable
  }
}

          \newcommand{\Nb}[1]{\mathcal{N}(#1)}   \newcommand{\degX}[1]{\deg_X\!\left(#1\right)}                  \newcommand{\Vstar}[2]{V^*_{#1}\!\left(#2\right)}          \newcommand{\prop}{\mathrm{prop}}
\newcommand{\usw}{\mathrm{usw}}
\newcommand{\EE}{\mathbb{E}}

\title{Beyond the Half-Approximation:\\Fair and Efficient Online Class Matching}

\author{Sander Borst\\
  Max Planck Institute for Informatics\\
  Saarbr\"ucken, Germany\\
  \texttt{sborst@mpi-inf.mpg.de}
  \and
  Max Springer\\
  Princeton University\\
  Princeton, NJ, USA\\
  \texttt{maxspringer@princeton.edu}
}

\date{}

\begin{document}

\maketitle
\begin{abstract}
    Online bipartite matching, where agents are known in advance but items arrive sequentially and must be irrevocably assigned, is fundamental to problems ranging from ride-sharing to online advertising. 
When agents belong to classes such as demographic groups or geographic regions, fairness demands equitable treatment across these groups. 
Recent work introduced class envy-freeness (CEF), a natural extension of the classical fair division notion: an algorithm is $\alpha$-CEF if each class receives value at least an $\alpha$ fraction of what it could extract from any other class's bundle. 
However, all known algorithms achieving constant-factor CEF guarantees attain utilitarian social welfare (total matching value) of at most $\frac{1}{2}$ times the optimum, far below the $1-\frac{1}{e} \approx 0.632$ achievable without fairness constraints.

We resolve the open question of whether fairness necessitates this efficiency loss, by introducing threshold-based algorithms parameterized by $\gamma \in [0,1]$ that equalize allocations across classes until threshold $\gamma$, then maximize efficiency. 
For divisible matching, this yields simultaneous $(1-e^{-\gamma})$-CEF and $(1 - \frac{e^{\gamma-1}}{\gamma+1})$-USW guarantees; for indivisible matching, $\frac{\gamma}{2}$-CEF with the same USW. 
Setting $\gamma > 0$ produces the first algorithms beating $\frac{1}{2}$-USW while maintaining constant CEF. 
We complement this with a novel upper bound construction, proving no non-wasteful $\alpha$-CEF algorithm can exceed $\frac{1 +\alpha - e^{\alpha-1}}{1+\alpha}$-USW and correcting prior bounds that were vacuous for $\alpha < 0.58$.
Our upper bound nearly matches our algorithms' performance, giving the first substantive characterization of the price of fairness in online class matching. \end{abstract}

\section{Introduction}
We study online bipartite matching under fairness constraints, a model motivated by real-time allocation problems such as ride-sharing and online advertising. In these settings, decisions must be made irrevocably as requests arrive: once a driver accepts a ride or an ad impression is shown, the allocation cannot be undone. Moreover, agents (drivers, advertisers) often belong to diverse demographic groups or geographic regions, and purely profit-maximizing algorithms may systematically disadvantage certain communities—for example, by allocating fewer rides to drivers from particular neighborhoods or fewer impressions to advertisers from underrepresented markets.

Formally, the online bipartite matching model considers a known set of agents and items that arrive sequentially in an adversarial order; upon each item's arrival, the algorithm must immediately and irrevocably match it to an agent.
Classical work in this area focuses solely on \emph{utilitarian social welfare} (USW), the total value of the matching. The celebrated RANKING algorithm of Karp, Vazirani, and Vazirani~\cite{karp1990optimal} achieves the optimal competitive ratio of $1-\frac{1}{e} \approx 0.632$ for this objective, and an equivalent guarantee holds for divisible matchings via water-filling~\cite{kalyanasundaram2000optimal}.

However, efficiency alone is insufficient when agents belong to protected or identifiable classes. 
Recent work by Hosseini et al.~\cite{hosseini_class_2024} introduced the \emph{online class matching} framework, where the $n$ agents are partitioned into $k$ known classes (representing demographic groups, geographic regions, or skill levels), and the algorithm must balance the two competing objectives of maximizing total value (USW) while ensuring equitable treatment across classes. 
They formalized \emph{class envy-freeness} (CEF) as the fairness criterion, wherein an algorithm is $\alpha$-CEF if each class receives value at least an $\alpha$ fraction of what it could extract from any other class's allocation. 
Intuitively, no class should strongly prefer another class's bundle of items. This essentially adds a non-convex constraint to the online matching problem, making it significantly more challenging than the classical setting.

\paragraph{Fairness-efficiency tension.} 
Hosseini et al.~\cite{hosseini_class_2024} further showed that non-wasteful algorithms (those that never leave an item unmatched when some agent desires it) can achieve $(1-\sfrac{1}{e})$-CEF for divisible matching and $\frac{1}{2}$-CEF for indivisible matching. 
Non-wastefulness is both natural (wasting capacity is undesirable) and structurally important.
Specifically, it directly implies a $\frac{1}{2}$-USW guarantee (Proposition~\ref{prop:nw-half-usw}). 
However, this creates a troubling gap: while the best possible USW without fairness constraints is $1-\sfrac{1}{e} \approx 0.632$, all known fair algorithms plateau at $\frac{1}{2}$ USW. 
This raises a fundamental question, posed explicitly as an open problem in~\cite{hosseini_class_2024}:

\begin{quote}
\emph{Can algorithms simultaneously achieve constant-factor CEF and USW strictly better than $\frac{1}{2}$?}
\end{quote}

Despite subsequent work establishing price of fairness bounds---upper limits on achievable USW given fairness constraints~\cite{hajiaghayi_fairness_2023}---no algorithm breaking the $\frac{1}{2}$ barrier was known. 
The central question remained: does fairness inherently force this severe efficiency loss?

\subsection{Our Contributions}
We resolve this open problem in the affirmative and provide the first parametric tradeoff between class envy-freeness and utilitarian welfare for online class matching.
Our main contributions are as follows.

\paragraph{1. Parametric algorithms achieving the full spectrum of tradeoffs.} 
We introduce threshold-based algorithms that smoothly interpolate between pure efficiency ($\gamma=0$) and strong fairness ($\gamma=1$) via a single parameter $\gamma \in [0,1]$.
Namely, the \textbf{Equal Filling Till Threshold} (EFTT) for divisible matching, and \textbf{Hybrid Ranking} for indivisible matching.
Both algorithms follow a simple two-phase design: until each class reaches threshold $\gamma$ saturation, distribute items equally across classes to enforce fairness; once all classes exceed $\gamma$, switch to efficiency-maximizing allocation (water-filling or RANKING-style). 
This clean separation enables simultaneous analysis of both objectives.

\begin{theorem}[Informal—Divisible Matching]\label{thm:eftt-informal}
    For every $\gamma \in [0,1]$, EFTT achieves $(1-e^{-\gamma})$-CEF and $\left(1-\frac{e^{\gamma-1}}{\gamma + 1}\right)$-USW.
\end{theorem}

\begin{theorem}[Informal—Indivisible Matching]\label{thm:hybrid-informal}
    For every $\gamma \in [0,1]$, Hybrid Ranking achieves $\frac{\gamma}{2}$-CEF and $\left(1-\frac{e^{\gamma-1}}{\gamma + 1}\right)$-USW.
\end{theorem}

Crucially, for any $\gamma \in (0,1)$, both algorithms beat $\frac{1}{2}$-USW while maintaining constant CEF, resolving the open question of~\cite{hosseini_class_2024}. 
For example, setting $\gamma=0.7915$ simultaneously yields $0.5468$-USW and $0.5468$-CEF for divisible matching, which is a substantial improvement over the previous $\frac12$ ceiling. 
The parameter $\gamma$ gives practitioners a mechanism to navigate the fairness-efficiency frontier based on application needs.

\paragraph{2. Tight price of fairness bounds.} 
We prove a substantially improved upper bound on what any non-wasteful fair algorithm can achieve:
\begin{theorem}[Informal—Price of Fairness]\label{thm:pof-informal}
    For divisible matchings, any randomized non-wasteful $\alpha$-CEF algorithm has USW competitive ratio at most $\frac{1+\alpha - e^{\alpha-1}}{1 + \alpha}$.
\end{theorem}

This corrects a critical deficiency in prior work. 
Hajiaghayi et al.~\cite{hajiaghayi_fairness_2023} established an upper bound of $\frac{1}{1+\alpha}$, but this bound exceeds $1-\sfrac{1}{e}$ (the absolute online matching limit) for $\alpha \le \frac{1}{e-1} \approx 0.58$, rendering it vacuous precisely where fairness matters most. 
Our bound correctly interpolates between $1-\sfrac{1}{e}$ (as $\alpha \to 0$, recovering the fairness-free limit) and $\frac12$ (as $\alpha \to 1$), providing meaningful constraints across all $\alpha \in (0,1]$. 
Figure~\ref{fig:pof} illustrates how our algorithms nearly match this bound.

\paragraph{Techniques and Broader Impact}
Our algorithmic technique of balancing fair distribution with efficiency maximization is conceptually simple yet powerful: enforce perfect fairness until the expected fractional load of a node reaches a target level $\gamma$, then switch to the optimal efficiency strategy.
The threshold parameter provides interpretability: system designers can explicitly set their fairness-efficiency priority rather than black-box tuning.

Our analytical technique pairs this algorithmic structure with a matching primal-dual framework. 
For EFTT, we construct an approximate dual solution whose definition mirrors the algorithm's phases via a carefully chosen potential function $g(z)$ that transitions at threshold $\gamma$. 
This enables simultaneous analysis of both USW and CEF, rather than analyzing each objective separately.
For Hybrid Ranking, we perform a primal-dual analysis based on a novel marking scheme.

Our price of fairness upper bound (Theorem~\ref{thm:pof-informal}) extends the classical adversarial upper-triangular construction~\cite{karp1990optimal} to a two-phase instance that jointly exploits CEF constraints and non-wastefulness. 
The first phase forces fair algorithms to equalize allocations (capturing fairness cost), while the second phase comprising $k-1$ copies of the upper-triangular instance exploits the resulting imbalance to limit achievable USW. This construction technique may prove useful for establishing price-of-fairness bounds in other constrained online problems.

\paragraph{Organization.}
Section~\ref{sec:rw} surveys related work on online matching, fair division, and price of fairness. Section~\ref{sec:prelim} formalizes the model and fairness notions. Sections~\ref{sec:disible} and~\ref{sec:indivis} present EFTT and Hybrid Ranking with full analysis. Section~\ref{sec:discuss} discusses extensions and open problems.

\usepgfplotslibrary{fillbetween}
\usetikzlibrary{patterns,arrows.meta}

\begin{figure}[t]
  \centering
  \begin{tikzpicture}
    \begin{axis}[
      width=0.78\textwidth,
      height=0.54\textwidth,
      clip=false,
      xlabel={$\alpha$-CEF},
      ylabel={$\beta$-USW},
      xlabel style={font=\small, yshift=2pt},
      ylabel style={font=\small, yshift=-4pt},
      xmin=0, xmax=1,
      ymin=0.50, ymax=0.70,
      xtick={0,0.5,1.0},
      ytick={0.50,0.5654,0.63212},
      yticklabels={$0.50$,$0.5654$,$1-\frac{1}{e}$},
      tick label style={font=\small},
      minor x tick num=1,
      minor y tick num=1,
      grid=major,
      major grid style={line width=0.35pt, draw=gray!25},
      axis line style={line width=0.7pt},
      legend style={
        at={(0.5,-0.18)},
        anchor=north,
        draw=gray!50,
        fill=white,
        fill opacity=0.95,
        text opacity=1,
        legend columns=2,
        font=\small,
        column sep=1em,
        row sep=0.25ex,
        rounded corners=2pt,
      },
      legend cell align={left},
      every axis plot/.append style={line cap=round},
      declare function={
        ub(\a)    = (\a - exp(\a-1) + 1)/(\a + 1);
        prior(\a) = 1/(\a + 1);
        shared(\g)= 1 - exp(\g-1)/(\g+1);
        efttx(\g) = 1 - exp(-\g);
        hybx(\g)  = \g/2;
      },
    ]

\addplot[
      name path=ourubshade,
      domain=0:1,
      samples=300,
      draw=none,
      forget plot
    ] {ub(x)};
    \path[name path=topaxis] (axis cs:0,0.7) -- (axis cs:1,0.7);
    \addplot[
      pattern=north east lines,
      pattern color=gray!60,
      draw=none,
      forget plot
    ] fill between[of=ourubshade and topaxis];

\addplot[
      domain=0:1,
      samples=300,
      line width=2.2pt,
      color={rgb,1:red,0.13;green,0.37;blue,0.71},
    ] ({efttx(x)},{shared(x)});
    \addlegendentry{Equal Filling Till Threshold (divs. ours)}

    \addplot[
      domain=0:1,
      samples=300,
      line width=2.2pt,
      color={rgb,1:red,0.50;green,0.13;blue,0.63},
    ] ({hybx(x)},{shared(x)});
    \addlegendentry{Hybrid Ranking (int. ours)}

\addplot[
      domain=0:1,
      samples=300,
      line width=2.0pt,
      color={rgb,1:red,1;green,0.10;blue,0.10},
      dashed,
      dash pattern=on 3pt off 3pt,
    ] {ub(x)};
    \addlegendentry{PoF Bound (ours)}

    \addplot[
      domain=0.43:1,
      samples=300,
      line width=1.2pt,
      color={rgb,1:red,0.65;green,0.80;blue,0.65},
      dashed,
      dash pattern=on 4pt off 3pt on 1pt off 3pt,
    ] {prior(x)};

    \addplot[
      domain=0:1,
      samples=2,
      line width=1.0pt,
      color=black!45,
      dashed,
      dash pattern=on 3pt off 2pt,
    ] {1 - 1/exp(1)};
    \addlegendentry{PoF Bound \cite{hajiaghayi_fairness_2023}}

\node[
      fill=white,
      inner sep=2pt,
      font=\footnotesize,
      align=center
    ] at (axis cs:0.5,0.67) {no online solution};

\draw[
      -{Latex[length=2.0mm]},
      very thick,
      black
    ] (axis cs:0.70,0.588) -- (axis cs:0.70,0.564);

    \node[
      fill=white,
      inner sep=1.5pt,
      font=\footnotesize,
      align=center
    ] at (axis cs:0.82,0.578) {improved PoF\\ bound};

\node[
      fill=white,
      inner sep=1pt,
      font=\footnotesize
    ] at (axis cs:0.75,0.6385) {USW upper bound};

\addplot[
      only marks,
      mark=*,
      mark size=2.2pt,
      mark options={
        fill=white,
        draw={rgb,1:red,0.13;green,0.37;blue,0.71},
        line width=1.0pt
      },
    ] coordinates {
      (0.0,0.6321)
      (0.5,0.5654)
    };

    \node[fill=white, inner sep=1pt, font=\scriptsize, align=center] (A) at (axis cs:0.2,0.65) {Online matching bound~\cite{karp1990optimal} \\ $\gamma=0$};
    \draw[-{Latex[length=1.4mm]}, thin] (A) -- (axis cs:0.01,0.6331);

\addplot[
      only marks,
      mark=*,
      mark size=2.2pt,
      mark options={
        fill=white,
        draw={rgb,1:red,0.50;green,0.13;blue,0.63},
        line width=1.0pt
      },
    ] coordinates {
      (0.5,0.5)
};

    \node[fill=white, inner sep=1pt, font=\scriptsize] (C) at (axis cs:0.35,0.52) {\textsc{Random}~\cite{hajiaghayi_fairness_2023}};
    \draw[-{Latex[length=1.4mm]}, thin] (C) -- (axis cs:0.495,0.501);

\addplot[
      only marks,
      mark=*,
      mark size=2.2pt,
      mark options={
        fill=white,
        draw={rgb,1:red,0.13;green,0.37;blue,0.71},
        line width=1.0pt
      },
    ] coordinates {
      (1-1/exp(1),0.5)
};

    \node[fill=white, inner sep=1pt, font=\scriptsize] (D) at (axis cs:0.77,0.51) {\textsc{Equal Filling}~\cite{hosseini_class_2024}};
    \draw[-{Latex[length=1.4mm]}, thin] (D) -- (axis cs:0.635,0.501);

    \draw[
      <->,
      >={Latex[length=2.0mm]},
very thick,
      black
    ] (axis cs:0.50,0.505) -- (axis cs:0.50,0.5615);

    \node[
      fill=white,
      inner sep=1.5pt,
      font=\footnotesize,
      align=center
    ] at (axis cs:0.6,0.53) {divisibility\\gap};

    \end{axis}
  \end{tikzpicture}
  \caption{
    Price-of-fairness trade-off between the USW and
    CEF approximate guarantees. 
    The hatched region above our upper bound denotes guarantees for which no online solution exists. 
    Solid curves are achieved by our algorithms; dashed curves are upper bounds
    (shorthand divs. and int. are for divisible and integral matchings respectively).
  }
  \label{fig:pof}
\end{figure}
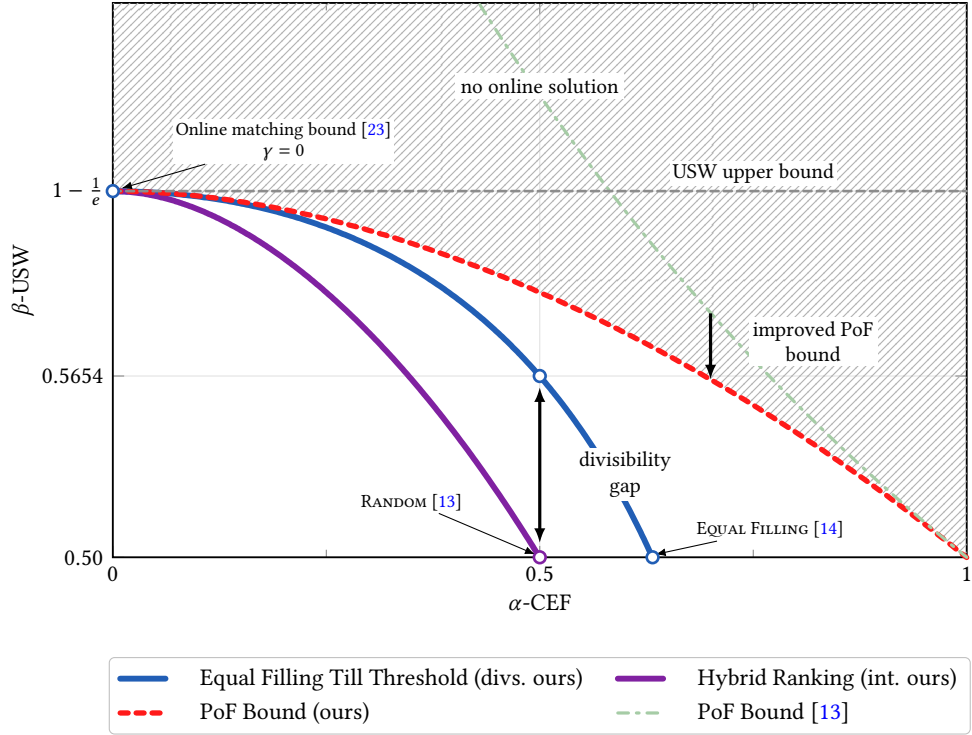

\section{Related Work} \label{sec:rw}

\paragraph{Online Matching.} 
The foundation for our efficiency benchmarks comes from the seminal work of Karp, Vazirani, and Vazirani~\cite{karp1990optimal}, who introduced the RANKING algorithm and proved it is $1-\sfrac{1}{e}$-competitive and optimal for adversarial online bipartite matching. 
For divisible matching, Kalyanasundaram and Pruhs~\cite{kalyanasundaram2000optimal} showed that water-filling achieves the same guarantee deterministically. 
These results establish a fundamental barrier: no online algorithm can exceed $1-\sfrac{1}{e}$ competitive ratio, regardless of computational power or randomization.
Online matching remains an active research area, with recent work focusing on generalizations of the bipartite matching problem, such as edge-weighted~\cite{blanc_multiway_2022-1, fahrbach_edge-weighted_2022}, vertex-weighted~\cite{huang2019online,jin2021improved} and non-bipartite matching \cite{gamlath_online_2019-2} as well as relaxed input models
such as random-order arrivals, \cite{mahdian2011online,karande2011online} and stochastic matching \cite{feldman2009online, huang2021online}.
For comprehensive surveys, we refer readers to \cite{mehta2013online} and \cite{huang_online_2024-6}.
Our work operates in the adversarial unweighted model, where $1-\sfrac{1}{e}$ remains the absolute efficiency ceiling.

\paragraph{Fair Division.} 
There is a vast literature on fair division in the offline setting.
For divisible goods with additive valuations, Varian~\cite{varian1973equity} showed that envy-free and Pareto-optimal allocations always exist, and Eisenberg and Gale~\cite{eisenberg1959consensus} provided efficient computation via market equilibria. 
For indivisible goods, exact envy-freeness is impossible in general. 
The community has developed two primary relaxations: envy-freeness up to one item (EF1)~\cite{lipton2004approximately}, which guarantees that envy can be eliminated by removing at most one item from the envied agent's bundle, and maximin share fairness (MMS)~\cite{budish2011combinatorial}, which ensures each agent receives at least their guaranteed share in a worst-case division. 
Caragiannis et al.~\cite{caragiannis2019unreasonable} proved EF1 allocations always exist under monotone valuations and are Pareto-optimal under additive valuations. 
While MMS allocations may not exist, polynomial-time approximations are known~\cite{garg2020improved, ghodsi2018fair, hosseini2022ordinal, hosseini2021guaranteeing, procaccia2014fair}.
Our setting differs fundamentally from this offline literature as we treat each class as an agent whose value is the sum of individual agent values, and we face irrevocable online decisions under adversarial arrivals. In the online setting fair division~\cite{aleksandrov2015online, benade2018make, gorokh2020online, walsh2011online, zeng2020fairness} has primarily focused on settings with additive valuations, where agents' values sum across items.

\paragraph{Fair Online Matching.} 
The online class matching framework was introduced by Hosseini et al.~\cite{hosseini_class_2024}, who formalized the fairness objectives we study: class envy-freeness (CEF), class proportionality (CPROP), and class maximin share (CMMS). 
They established strong results for deterministic algorithms in both divisible and indivisible settings. 
For divisible matching, they proved deterministic algorithms can achieve $(1-\sfrac{1}{e})$-CEF, which is optimal. For indivisible matching, they showed $\frac{1}{2}$-CEF is achievable deterministically and proved this is tight. 
For randomized algorithms, they achieved $0.593$-CPROP for indivisible matching via online correlated selection, though this algorithm provides no CEF guarantee.
Their work established that all known fair algorithms achieve at most $\frac{1}{2}$-USW due to the non-wastefulness constraint (Proposition~\ref{prop:nw-half-usw}). 
Crucially, they posed as an explicit open problem whether any algorithm could simultaneously achieve constant-factor CEF and USW strictly better than $\frac{1}{2}$. This question remained open despite subsequent work.

Yokoyama and Igarashi~\cite{yokoyamaasymptotic} recently made progress in the asymptotic regime, providing an algorithm that achieves CEF with high probability as the number of agents grows large. However, their result provides no finite-sample guarantees, no quantification of the USW tradeoff, and no parametric control over the fairness-efficiency spectrum.
Most recently, Hajiaghayi et al.~\cite{hajiaghayi_fairness_2023} made important contributions to understanding this landscape. They provided the \textsc{Random} algorithm for indivisible matching, which achieves simultaneous $\frac{1}{2}$-CEF, $\frac{1}{2}$-CPROP, and $\frac{1}{2}$-USW guarantees while maintaining non-wastefulness—the first algorithm to achieve any constant CEF with non-wastefulness. 
They complemented this with tight upper bounds showing no non-wasteful algorithm can achieve better than $\frac{e^2-1}{e^2+1} \approx 0.761$-CEF for indivisible matching or better than $0.677$-CEF for deterministic divisible matching.
However, their work left the central question unresolved: all their algorithms achieve exactly $\frac{1}{2}$-USW, matching the lower bound from non-wastefulness alone.

\section{Preliminaries} \label{sec:prelim}

We follow the online bipartite matching model of~\cite{hajiaghayi_fairness_2023,hosseini_class_2024}.
Consider a bipartite graph $G = (N, M, E)$ where $N$ is a known set of \emph{agents},
$M$ is a set of \emph{items}, and $E \subseteq \items \times \agents$.
Agent $a$ \emph{likes} item $o$ when $(o, a) \in E$.
The agents are partitioned into $k$ known \emph{classes} $N_1, \ldots, N_k$;
we write $[k] = \{1, \ldots, k\}$ and use ``class $i$'' for $N_i$ throughout.
Items arrive one at a time in an adversarial order.
Upon arrival of item $o$, the neighbourhood $\Nb{o}$ is revealed
and the algorithm must immediately and irrevocably
assign (a fraction of) $o$ to agents.

\paragraph{Matchings.}
A \emph{fractional matching} is a matrix
$X = (x_{o,a}) \in [0,1]^{M \times N}$ with
$\sum_{a} x_{o,a} \le 1$ for all $o \in M$ and
$\sum_{o} x_{o,a} \le 1$ for all $a \in N$.
An \emph{integral matching} further restricts $x_{o,a} \in \{0,1\}$;
these correspond to the \emph{divisible} and \emph{indivisible} settings, respectively.
We write $\degX{v} = \sum_{e \ni v} x_e$ for the fractional degree of vertex $v$,
and say agent $a$ is \emph{saturated} (resp.\ item $o$ is \emph{assigned})
when $\degX{a} = 1$ (resp.\ $\degX{o} = 1$).

\paragraph{Valuations.}
The value of agent $a$ is $V_a(X) = \sum_{o : (a,o) \in E} x_{o,a}$,
and the \emph{class value} is $V_i(X) = \sum_{a \in N_i} V_a(X)$.
For any matching $X$, let $y_i(X)\in [0,1]^{M}$ be defined as $y_i(X)_o = \sum_{a \in N_i} x_{o,a}$ for all $o \in M$.
For an integral matching $X$, let $Y_i(X)$ be the set of items assigned to agents in class $i$; then $y_i(X)$ is the incidence vector of $Y_i(X)$.
For a vector $y \in [0,1]^M$, the \emph{optimistic valuation} $\Vstar{i}{y}$
is the size of a maximum fractional matching between $N_i$ and $M$ that has fractional degree at most $y_o$ for each item $o$ and at most $1$ for each agent $a$.
For a set $S \subseteq M$, we write $\Vstar{i}{S}$ as shorthand for $\Vstar{i}{\mathbf{1}_S}$.
Crucially, $\Vstar{i}{\cdot}$ is
\emph{subadditive} and \emph{concave}.

\paragraph{Notions of Fairness \& Efficiency.}
We adopt the following class-level analogues of standard fair-division criteria.

\begin{definition}[$\alpha$-CEF]
For $\alpha \in (0,1]$, a matching $X$ is \emph{$\alpha$-class envy-free} ($\alpha$-CEF) if
$V_i(X) \ge \alpha \cdot \Vstar{i}{y_j(X)}$ for all $i,j \in [k]$.
\end{definition}

The \emph{utilitarian social welfare} is $\usw(X) = \sum_i V_i(X)$,
and $X$ is \emph{$\alpha$-USW} if $\usw(X) \ge \alpha \cdot \usw(X^*)$
for all matchings $X^*$.
A matching is \emph{non-wasteful} (NW) if no item is unassigned while some
unsaturated agent likes it; equivalently, NW matchings are maximal matchings.

\begin{proposition}[Folklore]
\label{prop:nw-half-usw}
Every non-wasteful matching is $\tfrac{1}{2}$-USW.
\end{proposition}

\paragraph{Randomized algorithms.}
For randomized algorithms, all guarantees hold in expectation over the algorithm's
randomness.
For instance, $\alpha$-CEF requires
$\EE[V_i(X)] \ge \alpha \cdot \EE[\Vstar{i}{y_{j}(X)}]$ for all $i,j$.
Note that since $\Vstar{i}{\cdot}$ is concave,
Jensen's inequality gives $\Vstar{i}{\EE[y_{j}(X)]} \ge \EE[\Vstar{i}{y_{j}(X)}]$,
so a randomized $\alpha$-CEF guarantee does not automatically imply a
deterministic fractional algorithm.
 
\section{The Price of Fairness in Divisible Matching} \label{sec:disible}
We first consider the divisible setting where arriving items can be fractionally matched to agents

\subsection{Divisible Algorithm}
We proceed to prove our main algorithmic result with Theorem~\ref{thm:eftt-informal}, formalized as follows.

\begin{theorem} \label{thm:divisible-algo}
    For every $\gamma \in [0,1]$, there exist an online divisible algorithm that achieves $\left( 1 - \frac{e^{\gamma-1}}{\gamma + 1} \right)$-USW and $(1-e^{-\gamma})$-CEF guarantees simultaneously.
\end{theorem}

Let $\gamma \in [0,1]$ define a threshold. Our algorithm continuously increases the matching, meaning that the mass of $o$ is allocated infinitesimally over time according to the current rule. This can be interpreted as a water-filling process, where levels rise continuously. 
Upon the arrival of item $o$, our \textbf{\underline{E}qual \underline{F}illing \underline{T}ill \underline{T}hreshold (EFTT)} algorithm proceeds in two phases.

\paragraph{Phase I: Equal Filling (up to threshold $\gamma$)}
Let $Z_\gamma$ be the set of classes containing at least one agent that likes $o$ and has fractional degree less than $\gamma$.
While $Z_\gamma$ is non-empty and $o$ is not fully assigned, continuously increase the matching in the following way:
\begin{enumerate}
    \item Distribute load of $o$ equally among all classes in $Z_\gamma$
    \item Within each class $i \in Z_\gamma$, allocate this mass to the agent that likes $o$ and has minimum fractional degree
    \item When any agent reaches fractional degree $\gamma$, if their class no longer has any agent liking $o$ with degree less than $\gamma$, remove that class from $Z_\gamma$.
\end{enumerate}

\paragraph{Phase II: Water Filling (beyond threshold $\gamma$)}
Once $Z_\gamma = \emptyset$, allocate mass of $o$ to the agent that likes $o$ with smallest fractional degree, regardless of class membership.

Naturally, the USW and CEF guarantees of this algorithm are dictated by the value of $\gamma$.
We proceed to derive this dependence.

\begin{lemma}
    The EFTT algorithm with threshold $\gamma$ is $\left( 1 - \frac{e^{\gamma-1}}{\gamma + 1} \right)$-USW.\label{lem:eftt-usw}
\end{lemma}
\begin{proof}
    Let $\delta = 1-\frac{e^{\gamma-1}}{\gamma + 1}$. The guarantee is proved via the standard fractional bipartite matching linear program. 
    Recall that $V = N \cup M$.
The primal and corresponding dual are:
\begin{center}
    \begin{minipage}[t]{0.45\textwidth}
        \centering
        \textbf{Primal (P)}
        \begin{align*}
            \max\quad & \sum_{e \in E} x_e \\
            \text{s.t.}\quad & \sum_{e \in \delta(v)} x_e \le 1 \quad && \forall v \in V, \\
            & x_e \ge 0 \quad && \forall e \in E.
        \end{align*}
    \end{minipage}
    \hspace{0.5em}\vrule\hspace{0.5em}\begin{minipage}[t]{0.45\textwidth}
    \centering
    \textbf{Dual (D)}
    \begin{align*}
        \min \quad & \sum_{v \in V} y_v \\
        \text{s.t.} \quad & y_o + y_a \ge 1 \quad && \forall e = (o,a) \in E, \\
        & y_v \ge 0 \quad && \forall v \in V.
    \end{align*}
    \end{minipage}
\end{center}
    By weak duality, for any feasible primal $x$ and feasible dual $y$, $\sum_e x_e \le \sum_v y_v$. A $\delta$-approximate dual solution for a feasible primal $x$ is a vector $y \ge 0$ such that $\sum_v y_v = \sum_e x_e$ and $y_o + y_a \ge \delta$ for every $e = (o,a) \in E$. Then $\sum_e x_e = \sum_v y_v \ge \delta \cdot |\mathrm{OPT}|$ (since $y/\delta$ is feasible for (\textbf{D}), so $\sum_v y_v \ge \delta \cdot \min(\textbf{D}) = \delta \cdot \max(\textbf{P})$). Thus the primal value is at least $\delta$ times the optimum, i.e. the algorithm is $\delta$-USW.

    It therefore suffices to construct a $\delta$-approximate dual solution: a vector $y \in \mathbb{R}^V_{\ge 0}$ that satisfies $\sum_{v \in V} y_v = \sum_{e \in E} x_e$ and $\sum_{v \in e} y_v \ge \delta \quad \forall e \in E$.
To construct such a vector, start with $y_v = 0$ for all $v\in V$. 
    While the fractional matching is continuously increasing, we also increase the dual solution.
    Define the following non-decreasing function $g$:
    \begin{align*}
        g(z)=\begin{cases}
            \frac{\exp(\gamma-1)}{\gamma +1}& \text{if } z\leq \gamma,\\
            \exp(z-1) & \text{otherwise}.
        \end{cases}
    \end{align*}
    While increasing the matching on edge $e=(o,a)$ at rate $r$, we increase $y_a$ at rate $g(\degX{a})r$ and $y_o$ at rate $(1-g(\degX{a}))r$.
Note that this continuous allocation procedure ensures the equality $\sum_{v\in V}y_v = \sum_{e\in E}x_e$ is maintained.
    Further observe that for $\degX{a}\in [0,1]$ we have $g(\degX{a})\in [0,1]$.
    Thus, each component of $y$ is non-negative.

    Now, consider $x$ and $y$ after the execution of the algorithm. Take an arbitrary edge $e=(o,a)$ with $o$ an online node. If $\degX{a}=1$, then we must have:
    \begin{align*}
        y_o+y_a \geq y_a \geq \int_0^1 g(z)dz = \gamma\cdot \frac{\exp(\gamma-1)}{\gamma +1} + 1 - \exp(\gamma-1) = \delta.
    \end{align*}

    If instead $\degX{a}< 1$ at termination, then we must have $\degX{o}= 1$, as our algorithm is non-wasteful.
Further, if $\degX{a}\in (\gamma, 1)$, then $o$ must have been matched exclusively to $a'$ with $\degX{a'}\leq \degX{a}$ at the time of matching, so the total dual mass assigned to $o$ from those matches is at least $1-g(\degX{a})$, as $g$ is non-decreasing. 
    This implies that $y_o\geq 1-g(\degX{a})$.
    Therefore, $y_a\geq \int_0^{\degX{a}}g(z)dz$ which implies:
    \begin{align*}
        y_o+y_a &\geq 1-g(\degX{a}) + \int_0^{\degX{a}}g(z)dz\\
&= 1-\exp(\degX{a}-1) +   \gamma\cdot \frac{\exp(\gamma-1)}{\gamma +1}+ \exp(\degX{a}-1) - \exp(\gamma-1)=\delta.
    \end{align*}
    
    Finally, if $\degX{a}\leq \gamma$, then $o$ must have been matched only to agents $a'$ with $\degX{a'}\leq \gamma$. 
    Hence, $y_o=1-g(\gamma)$, since $g$ is constant for $z\leq \gamma$. 
So we must have:
    \begin{align*}
        y_o+y_a \geq y_o = 1-g(\gamma)=\delta.
    \end{align*}
\end{proof}

\begin{lemma} \label{lem:eftt-cef}
    EFTT with threshold $\gamma$ is $(1-e^{-\gamma})$-CEF.
\end{lemma}
\begin{tproof}[Proof of Lemma~\ref{lem:eftt-cef}]
    Consider an arbitrary class $i$. Let $N_i(z)$ be the set of agents in class $i$ with fractional degree at least $z$. Let $f(z)=|N_i(z)|$. Let $\bar{N_i}(z)$ be the set of agents in class $i$ with fractional degree less than $z$. Observe that $V_i(x)=x(\delta(C_i))$, the fractional value of the matching assigned to class $i$, can now be written as:
    \begin{align*}
        V_i(x)=\int_0^1 f(x)dx.
    \end{align*}

    Now consider an arbitrary $j\neq i$ and an arbitrary $\theta\in [0,\gamma]$. For each object compatible with an agent $a\in \bar{N}_i(\theta)$, class $i$ will always be part of the set $Z_\gamma$ of classes with agents of fractional degree less than $\gamma$. Hence, we only execute the Equal Filling step for such objects, which ensures that the rate of increase for class $i$ is at least as large as the rate of increase for class $j$ for such objects.  
So these nodes must have contributed at most $\int_0^\theta f(z)dz$ to $V_i^*(y_j(X))$. Every node in $N_i(\theta)$ might contribute value at most $1$ to $V_i^*(y_j(X))$. So we have:
    \begin{align*}
        V_i^*(y_j(X))\leq \int_0^\theta f(z)dz + f(\theta).
    \end{align*}
    So, we get:
    \begin{align*}
        (1-\exp(-\gamma))\cdot V_i^*(y_j(X)) &= \int_0^\gamma \exp(\theta-\gamma)V_i^*(y_j(X))d\theta \\
        &\leq \int_0^\gamma \exp(\theta-\gamma)\left(\int_0^\theta f(z)dz + f(\theta)\right)d\theta.
    \end{align*}
    Swapping the order of integration in the first term, $\int_0^\gamma \exp(\theta-\gamma)\int_0^\theta f(z)dz\,d\theta = \int_0^\gamma f(z)\int_z^\gamma \exp(\theta-\gamma)d\theta\,dz = \int_0^\gamma f(z)(1-\exp(z-\gamma))dz$. The second term is $\int_0^\gamma \exp(\theta-\gamma)f(\theta)d\theta$. So the right-hand side equals
    \begin{align*}
        \int_0^\gamma f(\theta)(1-\exp(\theta-\gamma))d\theta + \int_0^\gamma f(\theta)\exp(\theta-\gamma)d\theta = \int_0^\gamma f(\theta)d\theta \leq \int_0^1 f(\theta)d\theta = V_i(x).
    \end{align*}
\end{tproof}

\subsection{Impossibility Result}
\label{subsec:impossibility}
We now present a new improved upper bound on the USW approximation guarantee of any non-wasteful $\alpha$-CEF algorithm for the divisible setting. This result relies on a careful construction of an adversarial instance that builds on two instances from the literature: the instance used by~\cite{hajiaghayi_fairness_2023} for showing an upper bound on the USW approximation guarantee of any non-wasteful $\alpha$-CEF algorithm, and the upper-triangular instance used by~\cite{karp1990optimal} for showing an upper bound on the approximation guarantee of any online algorithm for maximum matching. 
\begin{theorem}\label{thm:price-of-fairness}
    No randomized non-wasteful online algorithm with an $\alpha$-CEF guarantee for the divisible setting can achieve an approximation to the USW objective greater than $\frac{1 + \alpha - e^{\alpha-1}}{1 + \alpha} $. \label{thm:impossibility}
\end{theorem}
\begin{tproof}[Proof of \cref{thm:impossibility}]
    Suppose that an algorithm, $\mathcal A$, is $\alpha$-CEF for some $\alpha \in (0,1)$. 
    Let $p$ and $q$ be coprime integers such that $ p/q \leq  \alpha$.
    The proof proceeds by considering an instance with $k-1$ classes of $q$ agents, as well as a $k$-th class comprised of $q(k-1)$ agents. 
    The adversarial input stream consists of two phases.
    In the first phase, $p(k-1) + q$ items arrive, each of which has incident edges to every agent in the graph, whereas in the second phase, $k-1$ instances of the $q$-sized upper-triangular construction due to~\cite{karp1990optimal} specific to the $k-1$ smaller classes arrive sequentially.
    Let $X^{(t)}$ be the fractional matching after the arrival of the $t$-th item and let $c_i(t)=\Exp[V_i(X^{(t)})]$. Let $\tau = p(k-1) + q$. 
    We first have the following claim for the given instance due to~\cite{hajiaghayi_fairness_2023}.
    \begin{claim} \label{claim:nw-bound}
        For any non-wasteful $\alpha$-CEF algorithm and $\tau = p(k-1) + q$, we must have that $$\sum_{i=1}^{k-1} c_i(\tau) \ge p(k-1)$$
    \end{claim}
    We now describe the second phase of the instance in more detail and focus on any class $\class{i}$.
    Let $\class{i}^0:=N_i$ be the set of all agents in class $i$.
    In the second phase $q$ items will arrive which connect to the agents of $\class{i}^0$.  
    The first such item will connect to \emph{all} agents in $\class{i}^0$. 
    Now let $\class{i}^1= \class{i}^0 \setminus \{\arg\min_{a\in \class{i}^0} \Exp[V_a(X^{(\tau + 1)})]\}$. 
    Then the second item will connect to all nodes in $\class{i}^1$, and so on with $\class{i}^j = \class{i}^{j-1} \setminus \{\arg\min_{a\in \class{i}^{j-1}} \Exp[V_a(X^{(\tau + j)})]\}$. 
    The $j$-th item in the second phase will connect to all nodes in $\class{i}^{j-1}$.

    For the purpose of analysis, define $y_j:= \min_{a \in \class{i}^{j-1}}  \Exp[V_a(X^{(\tau + j)})]$.
    Let $z_j := \sum_{a \in \class{i}^{j-1}}  \Exp[V_a(X^{(\tau + j)})]$. 
    Observe that the total expected value of the matching of class $i$ is $\sum_{j=1}^{q} y_j$.
    Since the minimum is smaller than the average, we have $y_j \leq \frac{z_{j}}{|\class{i}^{j-1}|}$.
    Now observe that
    \begin{align*}
        z_{j} \leq c_i(\tau)+j-\sum_{s=1}^{j-1}y_s.
    \end{align*}
    and furthermore,
    \begin{align}
        y_{j} \leq \frac{z_j}{|\class{i}^{j-1}|}=\frac{z_j}{q-j+1}. \label{eq:y_eq}
    \end{align}
    
    In other words, $z_j$ and $y_j$ are a feasible solution to the following linear program with objective value equal to the total value of the fractional matching gained by the algorithm in both phases:

    \begin{center}
    \begin{minipage}[t]{0.45\textwidth}
    \centering
    \textbf{Primal}
    \begin{align*}
      \max \quad
        & \sum_{j=1}^{q} y_j \\
      \text{s.t.} \quad
        & \sum_{s=1}^{j-1} y_s + z_j \;\leq\; c_i(\tau) + j
          && \forall\, j \in [q], \\
        & y_j - \frac{z_j}{q - j + 1} \;\leq\; 0
          && \forall\, j \in [q], \\
        & y_j \;\leq\; 1
          && \forall\, j \in [q], \\
        & z_j \;\geq\; 0,\quad y_j \;\geq\; 0
          && \forall\, j \in [q].
    \end{align*}
    \end{minipage}\hspace{1em} \vrule \hspace{1em}
    \begin{minipage}[t]{0.45\textwidth}
    \centering
    \textbf{Dual}
    
    \begin{align*}
      \min \quad
        & \sum_{j=1}^{q} \bigl( (c_i(\tau) + j)\,\pi_j + \lambda_j \bigr) \\
      \text{s.t.} \quad
        & \sum_{j=s+1}^{q} \pi_j + \rho_s + \lambda_s \;\geq\; 1
          && \forall\, s \in [q], \\
        & \pi_s - \frac{\rho_s}{q - s + 1} \;\geq\; 0
          && \forall\, s \in [q], \\
        & \pi_s,\, \rho_s,\, \lambda_s \;\geq\; 0
          && \forall\, s \in [q].
    \end{align*}
    \end{minipage}
    \end{center}    
    Let $\psi = \lceil{q(1-\exp(c_i(\tau)/q - 1))}+1\rceil$. Note that $\psi \leq q$.
We show that $\rho_j = \frac{q-\psi}{q-j}\mathbf{1}_{j\leq \psi}$, $\pi_j=\rho_j - \rho_{j-1} = \frac{q-\psi}{(q-j)(q-j+1)}\mathbf{1}_{j\leq \psi}$ and $\lambda_s = \mathbf{1}_{s> \psi}$ is a feasible solution to the dual LP. 
    To check the first dual constraint, observe that for $s >\psi$ the constraint is trivially satisfied because $\lambda_s=1$ in that case. So, consider the case $s\leq \psi$. Here we have:
    \begin{align*}
        \sum_{j=s+1}^{q} \pi_j + \rho_s + \lambda_s &\geq \sum_{j=s+1}^{\psi}(\rho_{j}-\rho_{j-1}) + \rho_s =\rho_{\psi}=1.
    \end{align*}
    The second constraint is satisfied with equality for all $s\in [q]$. Hence, the dual solution is feasible. The objective value of the dual solution is:
    \begin{align*}
        \sum_{j=1}^{\psi} \left((c_i(\tau) + j) (\rho_j - \rho_{j-1})\right) + (q - \psi) &=\sum_{j=0}^{\psi - 1}(\rho_\psi -\rho_j) + c_i(\tau)(\rho_\psi-\rho_0)+(q-\psi)\\
        &= -\sum_{j=0}^{\psi - 1} \frac{q - \psi}{q - j}+ c_i(\tau) \psi/q  + q\\
        &\leq -\int_0^{\psi-1} \frac{q - \psi}{q - x} dx+ c_i(\tau) \psi/q  + q\\
        &= -(q - \psi) (\log(q) - \log(q - \psi + 1)) + c_i(\tau) \psi/q  + q\\
        &= (q - \psi) (\log(1 - \psi/q + 1/q)) + c_i(\tau) \psi/q  + q\\
        &\leq (q - \psi) (\log(\exp(c_i(\tau)/q - 1))) + c_i(\tau) \psi/q  + q\\
        &\leq (q - \psi) (c_i(\tau)/q-1) + c_i(\tau) \psi/q  + q\\
        &=c_i(\tau) + \psi.
    \end{align*}
    Hence, the objective value of the dual program is at most $c_i(\tau)  + \psi \leq c_i(\tau) + q(1 - \exp(c_i(\tau)/q - 1)) + 2$. 
    By weak duality, this also is an upper bound on the value of the primal program. 
    Summing over all classes $i\in [k]$ gives:
    \begin{align*}
        &\left(p(k - 1) + q - \sum_{i=1}^{k-1} c_i(\tau)\right) + \sum_{i=1}^{k-1} \left(c_i(\tau) + q(1 -\exp(c_i(\tau)/q - 1)) + 2\right)\\
        &= p(k - 1) + q  + \sum_{i=1}^{k-1} \left(q(1 - \exp(c_i(\tau)/q-1)) + 2\right)
    \end{align*}
    Let $\Sigma = \sum_{i=1}^{k-1} c_i(\tau)$. 
    Since the above function is symmetric and concave in the $c_i(\tau)$'s, it is maximized when all $c_i(\tau)$'s are equal. 
    So, we can set $c_i(\tau) = \frac{\Sigma}{k-1}$ for all $i\in [k-1]$. 
    By Claim~\ref{claim:nw-bound} we know that $\Sigma \geq p(k - 1)$. 
    This gives the following upper bound on the value of the matching:
    \begin{align*}
        p(k-1)+q  + \sum_{i=1}^{k-1} \left(q(1-\exp(p/q-1)) + 2\right)
    \end{align*}
    This makes the competitive ratio at most:
    \begin{align*}
        \frac{p(k-1)+q  + (k-1) \left(q(1-\exp(p/q-1)) + 2\right)}{p(k-1) + q + (k-1)q}=\frac{\frac{p}{q} +  \frac{1}{k-1}+(1-\exp(p/q-1)) + \frac{2}{q}}{\frac{p}{q}+\frac{1}{k-1}+1}
    \end{align*}
    Letting $k \to \infty$ gives
    \begin{align*}
        \frac{\frac{p}{q} +(1-\exp(p/q-1)) + \frac{2}{q}}{\frac{p}{q}+1}.
    \end{align*}
    and subsequently taking $q\to \infty$ and $p/q\to \alpha$ from below yields the final result.
\end{tproof}

\section{The Price of Fairness in Indivisible Matching} \label{sec:indivis}
We now consider the more challenging problem setting of \emph{indivisible} matchings.

\subsection{Hybrid Ranking}
Our main algorithmic result is formalized as follows.

\begin{theorem} \label{thm:indivisible-alg}
    For every $\gamma \in [0,1]$, there exist an online indivisible algorithm that achieves $\left( 1 - \frac{e^{\gamma-1}}{\gamma+1} \right)$-USW and $\frac{\gamma}{2}$-CEF guarantees simultaneously.
\end{theorem}

We consider the following algorithm that allows us to trade off between USW-competitiveness and CEF-guarantees. We call this approach \emph{hybrid ranking}.
Our approach is inspired by the Equal-Filling-Till-Threshold algorithm and is parameterized by a threshold $\gamma \in [0,1]$.
By decreasing the threshold, we increase the USW-competitiveness of the algorithm, while decreasing CEF-guarantees. 

This approach works as follows: Assume our set of agents $\agents$ is ordered. 
For agent $a\in \agents$ sample a random value $\mu_a$ uniformly from $[0,1]$.
Let $\agents'$ be the set of agents with $\mu_a \le \gamma$. 
On the arrival of an item $o$, we proceed as follows:
\begin{enumerate}
    \item If $o$ is compatible with some unmatched agents in $\agents'$: let $Z$ be the set of classes that contain such agents. We uniformly at random select a class $i$ from $Z$ and match $o$ to the unmatched agent in class $i$ that is in $N'$ and occurs first in the ordering of $\agents$.
    \item Otherwise, if $o$ is compatible with some unmatched agents in $\agents\setminus \agents'$: we match $o$ to the unmatched agent in $\agents\setminus \agents'$ with the lowest value of $\mu_a$.
\end{enumerate}

\begin{lemma}
    The above algorithm is $\left( 1 - \frac{e^{\gamma-1}}{\gamma + 1} \right)$-USW. \label{lem:hybrid-usw}
\end{lemma}
\begin{tproof}[Proof of \cref{lem:hybrid-usw}]
    Let $\delta = 1-\frac{e^{\gamma-1}}{\gamma + 1}$. 
    The guarantee is proved via the standard fractional bipartite matching linear program (see Lemma~\ref{lem:eftt-usw}).

    It suffices to construct a $\delta$-approximate dual solution: a vector $y \in \mathbb{R}^V_{\ge 0}$ with $\sum_{v \in V} y_v = \sum_{e \in E} \Exp[x_e]$ and
    \begin{align*}
        \sum_{v \in e} y_v \ge \delta \quad \forall e \in E.
    \end{align*}
    We will construct this probabilistically: Start with $y_v' = 0$ for all $v\in \agents \cup \items$. 
    Define the following nondecreasing function $g$:
    \begin{align*}
        g(z)=\begin{cases}
            \frac{\exp(\gamma-1)}{\gamma +1}& \text{if } z\leq \gamma,\\
            \exp(z-1) & \text{otherwise}.
        \end{cases}
    \end{align*}
    Whenever we match an item $o$ to an agent $a$, we set $y_a':=g(\mu_a)$ and $y_o':=1-g(\mu_a)$. Note that this way we will always have $\sum_{v\in V}y_v' = \sum_{e\in E}x_e$, since we only assign dual mass when we match an edge, and the total dual mass assigned to the endpoints of a matched edge is always $1$. 
Also, observe that for $\mu_a\in [0,1]$ we have $g(\mu_a)\in [0,1]$, so that each component of $y$ is non-negative.

    Now, we set $y_v = \Exp[y_v']$ for each $v\in V$. We will show that this is a $\delta$-approximate dual solution. Consider an arbitrary edge $e=(o,a)$. We will condition on the value of $\mu_{a'}$ for all agents $a'\neq a$ and denote this conditioning by $\mathcal{F}_{-a}$.
    Let $\theta$ be the threshold such that $o$ would be matched to some agent with $\mu_{a'} = \theta$ if $\mu_a$ were set to 1.
If $o$ is unmatched in that case, set $\theta:= 1$.  
    If $\theta \leq \gamma$, then we must have:
    \begin{align*}
        \Exp[y'_a+y'_o|\mathcal{F}_{-a}]\geq \Exp[y'_o|\mathcal{F}_{-a}] \geq 1-g(\theta) = 1-\frac{\exp(\gamma-1)}{\gamma + 1}=\delta.
    \end{align*}

    If $\theta > \gamma$, then  $a$ would be matched as long as $\mu_a \leq \theta$.
    So, we have:
    \begin{align*}
        \Exp[y'_a|\mathcal{F}_{-a}]\geq \int_0^\theta g(z)dz .
    \end{align*}
    If we have $\theta < 1 $, then $o$ will always be matched to some agent $a'$ with $\mu_{a'}\leq \theta$. Furthermore, we have $g(1)=1$, so that:
    \begin{align*}
        \Exp[y'_o|\mathcal{F}_{-a}] \geq 1-g(\theta).
    \end{align*}
    Combining the above, we get for $\theta \geq \gamma$.
    \begin{align*}
        \Exp[y'_a+y'_o|\mathcal{F}_{-a}] &\geq 1-g(\theta) + \int_0^\theta g(z)dz \\&=  1 - \exp(\theta-1) + \gamma \cdot \frac{\exp(\gamma-1)}{\gamma +1} + \exp(\theta-1) - \exp(\gamma-1) \\
        &= 1 + \exp(\gamma -1) \cdot (\frac{\gamma}{\gamma + 1} -1)=\delta .
    \end{align*}
So in all cases, we have $\Exp[y'_a+y'_o|\mathcal{F}_{-a}] \geq \delta$. By taking the expectation over $\mathcal{F}_{-a}$, we get $y_a+y_o \geq \delta$ as desired.
\end{tproof}

\begin{lemma}
    The above algorithm is $\frac{\gamma}{2}$-CEF. 
    \label{lem:hybrid-cef}
\end{lemma}

Let $\alpha = \gamma/2$ and fix $i$ and $j$.
To prove the lemma, we need to show that $\Exp[V_i(X)] \geq \alpha \cdot \Exp[V_i^*(Y_j(X))]$. Note that $V_i^*(Y_j(X))$ is upper bounded by the value of the following linear relaxation, where we use $E_{i,j}$ to denote the set of edges between items in $Y_j(X)$ and agents in class $i$:

    \begin{center}
    \begin{minipage}[t]{0.45\textwidth}
    \centering        \textbf{Primal}
        \begin{align*}
            \max_{x\in \mathbf{R}^{E_{i,j}}}\quad & \sum_{e \in E_{i,j}} x_e \\
            \text{s.t.}\quad & \sum_{e \in \delta(v)} x_e \le 1 \quad && \forall v \in V, \\
            & x_e \ge 0 \quad && \forall e \in E_{i,j}.
        \end{align*}
    \end{minipage}
    \hspace{1em}\vrule\hspace{1em}\begin{minipage}[t]{0.45\textwidth}
    \centering    \textbf{Dual}
    \begin{align*}
        \min \quad & \sum_{v \in V} y_v \\
        \text{s.t.} \quad & y_u + y_v \ge 1 \quad && \forall e = (u,v) \in E_{i,j}, \\
        & y_v \ge 0 \quad && \forall v \in V.
    \end{align*}
    \end{minipage}
\end{center}

We will construct a feasible dual solution $y$ such that $\Exp[\sum_{v \in V} y_v] \leq \frac{1}{\alpha} \cdot \Exp[V_i(X)]$, which will imply the desired guarantee. We do this by starting with $y_v = 0$ for all $v \in V$ and then \emph{marking} some agents in class $i$ and some items, which means that we set $y_v' = 1$. We now give an equivalent description of the algorithm that incorporates this marking step.

On the arrival of an item $o$, we proceed as follows:
\begin{enumerate}
    \item Set $W:=[k]$.
    \begin{enumerate}
        \item Choose a random class $i'$ from $W$ and remove it from $W$.
 \item  If $i'=i$, repeatedly keep marking the first unmarked agent in class $i$ compatible with $o$ (wrt the order of $\agents$). As soon as we either have marked an unmatched agent in $N'$ or we have marked all compatible agents in class $i$, continue to the next step.
        \item If $i'=j$ and there is an unmarked agent in class $i$ compatible with $o$, then mark $o$.
        \item If class $i'$ contains some unmatched agent in $\agents'$ compatible with $o$, we match $o$ to the first such agent in class $i'$ (wrt the order of $\agents$) and stop.
        \item Otherwise, if $W$ is not empty, we repeat the above steps, starting from step (a). If $W$ is empty, we proceed to step (2).
    \end{enumerate}
    \item We match $o$ to the compatible unmatched agent in $\agents\setminus \agents'$ with the lowest value of $\mu_a$.
\end{enumerate}
Observe that steps 1(b) and 1(c) are just used for the construction of the dual solution and do not affect the matching decisions of the algorithm.
We now show that the above algorithm constructs a feasible dual solution $y$ such that $\alpha \cdot  \Exp[\sum_{v \in V} y_v] \leq  \Exp[V_i(X)]$ via a sequence of claims.

First, we show that the marking procedure defines a valid dual certificate for the optimistic rematching problem defining $V_i^*(Y_j(X))$. 
In other words, every item-agent edge that class $i$ could use to envy class $j$'s bundle is covered either by marking the item or by marking the relevant class-$i$ agent.
\begin{claim}
    The above solution is feasible to the dual linear program. \label{claim:hybrid-cef-solution-feasible}
\end{claim}
\begin{tproof}[Proof of Claim~\ref{claim:hybrid-cef-solution-feasible}]
    Consider an edge $e=(o,a) \in E_{i,j}$. This means that $o\in Y_j(X)$, and hence 1(c) has been executed with $i'=j$.
Therefore, either $o$ is marked, or all agents in class $i$ compatible with $o$ are marked. That is, either $y_o=1$ or $y_a=1$, so that we have $y_o+y_a \geq 1$ as desired.
\end{tproof}

Having established feasibility, weak duality reduces the CEF guarantee to bounding the cost of the constructed dual. 
We thus proceed to show that the expected dual cost charged to class $i$ is controlled by the actual value class $i$ receives, up to the desired factor $1/\alpha$.
\begin{claim}
    We have $\Exp[\sum_{v \in V} y_v] \leq \frac{1}{\alpha} \cdot \Exp[V_i(X)]$. \label{claim:hybrid-cef-solution-value}
\end{claim}
We will prove this claim by showing that for each arrival of an item $o$, the increase of the left-hand side is bounded by the increase of the right-hand side. So, fix an item $o$ and condition on the behavior of the algorithm in steps 1(a)-(e) before the arrival of $o$ and on all rank values outside class $i$. Denote this conditioning as $\mathcal{F}_o$. It is important to note that under this conditioning, the values of $\mu_a$ for all unmarked agents in class $i$ are still independent and uniformly distributed in $[0,1]$. This holds because we only consult the value of $\mu_a$ for an agent $a\in \class{i}$ (in step 1(b), 1(d) or 2) after marking it (in step 1(b)).

Let $z$ be the number of unmarked
agents in class $i$ compatible with $o$ under the conditioning $\mathcal{F}_o$.
Let $R$ be the set of classes $\ell \neq i$ that contain an unmatched agent in $\agents'$ compatible with $o$ under $\mathcal{F}_o$.
Let $p$ be the probability, over the fresh random class order used for item $o$, that class $i$ is selected before every class in $R$.
Equivalently, $p=1/(|R|+1)$.
By symmetry, the probability that class $j$ is selected before every class in $(R\cup\{i\})\setminus \{j\}$ is at most $p$: it is exactly $p$ if $j\in R$, and it is $1/(|R|+2)\le p$ otherwise.
We now lower bound the increase in $\Exp[V_i(X)]$ in terms of $z$ and $p$.
\begin{claim}
    The increase in $\Exp[V_i(X)]$ is at least $(1-(1-\gamma)^z)p$. \label{claim:hybrid-cef-lhs-bound}
\end{claim}
\begin{tproof}[Proof of Claim~\ref{claim:hybrid-cef-lhs-bound}]
The increase in $\Exp[V_i(X)]$ is at least equal to the probability that $o$ is matched to an agent in $\class{i}$ in step 1(d).
Let $A_i$ be the event that at least one of the $z$ unmarked agents in class $i$ compatible with $o$ is in $\agents'$.
Since the rank values of these agents are independent and uniformly distributed in $[0,1]$ under $\mathcal F_o$, we have $\Pr[A_i]=1-(1-\gamma)^z$.
Moreover, $A_i$ is independent of the fresh random class order for item $o$.
Any unmarked compatible class-$i$ agent in $\agents'$ is unmatched: step 2 never matches agents in $\agents'$, and whenever such an agent is matched in step 1(d), step 1(b) marks it before the match occurs.
Thus, if $A_i$ occurs and class $i$ is selected before every class in $R$, then $o$ is matched to class $i$ in step 1(d).
The latter event has probability $p$, so the increase in $\Exp[V_i(X)]$ is at least $(1-(1-\gamma)^z)p$ as desired.
\end{tproof}

To further upper bound the increase in $\Exp[\sum_{v \in V} y_v]$ in terms of $z$ and $p$, observe that only agents in class $i$ and item vertices can be marked, so we write $\sum_{v \in V} y_v= \sum_{a \in \class{i}} y_a + \sum_{o\in M} y_o$. 
\begin{claim}
    The increase in $\Exp[\sum_{a \in \class{i}} y_a]$ is at most $p\cdot (1-(1-\gamma)^z)/\gamma$. \label{claim:hybrid-cef-rhs-bound-1}
\end{claim}
\begin{tproof}[Proof of Claim~\ref{claim:hybrid-cef-rhs-bound-1}]
    We only increase $\sum_{a \in \class{i}} y_a$ when we mark an agent in class $i$ in step 1(b). This can only happen if class $i$ is selected before every class in $R$, which happens with probability $p$. If we select class $i$, then the probability that we mark the $t$-th unmarked agent in $\class{i}$ is equal to the probability that all $t-1$ prior unmarked agents in class $i$ have $\mu_a > \gamma$. This happens with probability $(1-\gamma)^{t-1}$. Since the class-order event is independent of the remaining rank values in class $i$, the expected increase in $\sum_{a \in \class{i}} y_a$ is at most $p\cdot \sum_{t=1}^z (1-\gamma)^{t-1} = p\cdot (1-(1-\gamma)^z)/\gamma$ as desired.
\end{tproof}
 Lastly, we bound the item-side dual cost.
\begin{claim}
    The increase in $\Exp[\sum_{o\in M} y_o]$ is at most $p\mathbf{1}_{z\geq 1}$.
    \label{claim:hybrid-cef-rhs-bound-2}
\end{claim}
\begin{tproof}[Proof of Claim~\ref{claim:hybrid-cef-rhs-bound-2}]
    We only increase $\sum_{o\in M} y_o$ when we mark $o$ in step 1(c).
    This only happens when there is an unmarked agent in class $i$ compatible with $o$ (i.e. $z\geq 1$).
    Furthermore, $o$ can only be marked if class $j$ is selected before class $i$ and before every class in $R\setminus\{j\}$; otherwise the item either has already stopped at another active class, or class $i$ has already been processed.
    By the symmetry noted above, this event has probability at most $p$.
    Hence, the expected increase in $\sum_{o\in M} y_o$ is at most $p\mathbf{1}_{z\geq 1}$ as desired.
\end{tproof}
 \begin{proof}[Proof of Claim~\ref{claim:hybrid-cef-solution-value}]
    From Claims~\ref{claim:hybrid-cef-rhs-bound-1} and \ref{claim:hybrid-cef-rhs-bound-2}, we get that the increase in $\Exp[\sum_{v \in V} y_v]$ is at most $p\cdot (1-(1-\gamma)^z)/\gamma + p\mathbf{1}_{z\geq 1}$. By Claim~\ref{claim:hybrid-cef-lhs-bound}, the increase in $\Exp[V_i(X)]$ is at least $(1-(1-\gamma)^z)p$. We want to show that $\alpha$ times the former is at most the latter. That is, we want to show that for all $z\geq 0$:
    \begin{align*}
        \alpha \cdot \left(p\cdot (1-(1-\gamma)^z)/\gamma + p\mathbf{1}_{z\geq 1}\right) \leq (1-(1-\gamma)^z)p.
    \end{align*}
    Since for $z=0$ both sides are $0$, we can restrict to $z\geq 1$. By dividing both sides by $\alpha p (1-(1-\gamma)^z)$ and substituting $\alpha= \frac12\gamma$, the inequality is equivalent to showing that for all $z\geq 1$:
    \begin{align*}
        1/\gamma + \frac{1}{1-(1-\gamma)^z} \leq 2/\gamma.
    \end{align*}
    Since $1-(1-\gamma)^z \geq \gamma $, the above inequality holds, which implies that $\alpha$ times the increase in $\Exp[\sum_{v \in V} y_v]$ is at most the increase in $\Exp[V_i(X)]$ as desired.
     By summing over all items, we get $\alpha \Exp[\sum_{v \in V} y_v] \leq  \Exp[V_i(X)]$.
\end{proof}
 
Thus, we obtain the desired CEF guarantee by weak duality.
\begin{proof}[Proof of Lemma~\ref{lem:hybrid-cef}]
    By Claims~\ref{claim:hybrid-cef-solution-feasible} and \ref{claim:hybrid-cef-solution-value}, $y$ is a feasible dual solution with $\Exp[\sum_{v \in V} y_v] \leq \frac{1}{\alpha} \cdot \Exp[V_i(X)]$. By weak duality, we have $\alpha \Exp[V_i^*(Y_j(X))] \leq \alpha \Exp[\sum_{v \in V} y_v] \leq   \Exp[V_i(X)]$ as desired.
\end{proof}

\section{Discussion \& Open Questions} \label{sec:discuss}
We resolve the open question of whether constant-factor CEF can be achieved alongside USW strictly better than $\frac{1}{2}$, introducing the first algorithms that continuously parametrize the fairness-efficiency frontier in online class matching. 
The threshold-based design is simple, interpretable, and is analyzed via a clean primal-dual framework.
For divisible matching, EFTT achieves $(1-e^{-\gamma})$-CEF and $(1-\frac{e^{\gamma-1}}{\gamma+1})$-USW; for indivisible matching, Hybrid Ranking achieves $\frac{\gamma}{2}$-CEF with the same USW. 
While our focus is on CEF, our techniques might also extend to other group-fairness notions, such as group proportionality and maximin share fairness, which we leave for future work. In \cref{sec:apx-cprop} we provide preliminary evidence of this potential by showing that rounding EFTT to an indivisible solution via online correlated selection yields proportionality guarantees parametrized by $\gamma$, generalizing prior results from \cite{hosseini_class_2024}.

Key open problems include: (1) closing the gap between our algorithms and upper bound to characterize optimality, (2) improving the $\frac{\gamma}{2}$-CEF guarantee for indivisible matching, possibly through new rounding techniques beyond online correlated selection, and (3) extending the parametric framework to richer valuation classes (weighted preferences) and alternative arrival models (random-order, stochastic).

\section*{Acknowledgements}
We thank Suho Shin for helpful discussions and his feedback on an early draft of this work.

\bibliographystyle{plain}
\bibliography{references}

\clearpage
\appendix

\section{Appendix}
\subsection{Class Proportionality} \label{sec:apx-cprop}
We now extend the analysis of EFTT to the indivisible setting, showcasing its integrality guarantees after rounding via online correlated selection (OCS)~\cite{fahrbach_edge-weighted_2022}.
While the rounding procedure does not readily provide guarantees on envy-freeness, we can instead analysis population level fairness via proportionality.
More formally, the \emph{proportional share} of class $i$ is
$\prop_i = \max_{X} \min_{j \in [k]} \Vstar{i}{y_j(X)}$,
where the maximum is over all fractional matchings.

\begin{definition}[$\alpha$-CPROP]
A matching $X$ is \emph{$\alpha$-class proportional} ($\alpha$-CPROP) if
$V_i(X) \ge \alpha \cdot \prop_i$ for every class $i \in [k]$.
\end{definition}

We round the fractional solution produced by EFTT into an indivisible matching using online correlated selection (OCS), following~\cite{hosseini_class_2024}. 
The algorithm maintains a guiding fractional matching $\tilde{x}_{a,o}$ by running EFTT (without capping agent values at $1$).
When each item $o$ arrives, we obtain the vector $(\tilde{x}_{a,o})_{a \in N}$ with $\sum_a \tilde{x}_{a,o} \le 1$, feed it to the OCS, and match $o$ to the selected agent if that agent is unmatched, and otherwise to an arbitrary unmatched agent who likes $o$ (for non-wastefulness).

We use the following guarantee.
\begin{lemma}[OCS; see~\cite{fahrbach_edge-weighted_2022}]\label{lem:ocs}
    There is a polynomial-time online algorithm which, in each step, takes as input a non-negative vector $(\tilde{x}_{a,o})_{a \in N}$ for some $o \in M$ satisfying $\sum_{a \in N} \tilde{x}_{a,o} \le 1$ and selects an agent $a$ with positive $\tilde{x}_{a,o}$. By the end, each agent $a$ is selected at least once with probability at least
    \[
        p(\tilde{x}_a) = 1 - \exp \left( -\tilde{x}_a - \frac12 \tilde{x}_a^2 - \frac{4 - 2\sqrt{3}}{3} \tilde{x}_a^3 \right),
    \]
    where $\tilde{x}_a = \sum_{o \in M} \tilde{x}_{a,o}$.
\end{lemma}

As we will see in the proof of our main theorem, OCS provides a marginal lower bound on the probability that each agent is selected as a function of its guiding fractional load.
While these bounds are sufficient for proportionality (whose bounding value is independent of the rounded allocation), correlations among the items associated to another class $j$ can substantially alter the envious relationships between classes.
The CEF proofs in the main text require controlling pairwise relations and accumuluated correlations prevent such guarantees with OCS rounding.

We proceed with our main rounding result.

\begin{theorem}
\label{thm:rounded-cprop}
Fix $\gamma \in [0,1]$. Let $\widetilde X = (\widetilde x_{a,o})$ be the
\emph{uncapped} guiding fractional process obtained by running EFTT with
threshold $\gamma$, and let $X$ be the integral matching produced by the
OCS rounding scheme: when item $o$ arrives, feed the vector
$(\widetilde x_{a,o})_{a \in N}$ to OCS, match $o$ to the selected agent if
that agent is unmatched, and otherwise match $o$ to an arbitrary unmatched
agent who likes it (if one exists).

Then, for every class $i \in [k]$,
\[
    \EE[V_i(X)] \;\ge\; \rho(\gamma)\cdot \prop_i,
\]
where
\[
    \rho(\gamma)
    \;=\;
    \int_0^\gamma e^{-\theta} p'(\theta)\, d\theta,
    \qquad
    p(t)
    \;=\;
    1 - \exp\!\left(
        -t - \frac{1}{2}t^2 - \frac{4-2\sqrt{3}}{3}t^3
    \right).
\]
Hence the algorithm is $\rho(\gamma)$-CPROP in expectation.
\end{theorem}

\begin{proof}[\cref{thm:rounded-cprop}]
Fix a class $i \in [k]$. For each agent $a \in N$, write
\[
    \widetilde x_a := \sum_{o \in M} \widetilde x_{a,o},
\]
so $\widetilde x_a$ is the total (uncapped) guiding load of agent $a$.

For $\theta \ge 0$, define
\[
    H_i(\theta) := \{a \in N_i : \widetilde x_a \ge \theta\},
    \qquad
    L_i(\theta) := N_i \setminus H_i(\theta),
    \qquad
    f_i(\theta) := |H_i(\theta)|.
\]
Also define
\[
    A_i(\theta)
    := \int_0^\theta f_i(z)\,dz
    = \sum_{a \in N_i} \min\{\widetilde x_a,\theta\}.
\] 

We first lower bound the expected value obtained by class $i$ in the rounded
matching $X$.

By Lemma~\ref{lem:ocs}, each agent $a$ is selected by OCS at least once with
probability at least $p(\widetilde x_a)$. Since any agent selected at least
once is certainly matched by the end of the algorithm, we have
\[
    \EE[V_i(X)]
    \;\ge\;
    \sum_{a \in N_i} p(\widetilde x_a).
\]
Using the layer-cake identity and Fubini's theorem,
\[
    \sum_{a \in N_i} p(\widetilde x_a)
    =
    \sum_{a \in N_i} \int_0^{\widetilde x_a} p'(\theta)\,d\theta
    =
    \int_0^\infty p'(\theta)\, f_i(\theta)\, d\theta.
\]
Therefore,
\begin{equation}
\label{eq:ocs-lower-bound}
    \EE[V_i(X)]
    \;\ge\;
    \int_0^\infty p'(\theta)\, f_i(\theta)\, d\theta.
\end{equation}

It remains to upper bound $\prop_i$ in terms of $f_i(\theta)$ and $A_i(\theta)$,
for $\theta \le \gamma$.

\medskip
\noindent
\textbf{Claim.}
For every $\theta \in [0,\gamma]$,
\begin{equation}
\label{eq:prop-envelope}
    \prop_i \;\le\; A_i(\theta) + f_i(\theta).
\end{equation}

\medskip
\noindent
\emph{Proof of the claim.}
Fix $\theta \in [0,\gamma]$.
Consider an arbitrary fractional matching
$Z = (z_{a,o}) \in [0,1]^{N \times M}$.
For each class $j \in [k]$, let
\[
    \widehat y^j_o := \sum_{a \in N_j} z_{a,o}
    \qquad (o \in M).
\]
Then $\widehat y^j \in [0,1]^M$ and
$\sum_{j=1}^k \widehat y^j_o \le 1$ for every item $o$.
Since $\prop_i$ is defined by maximizing the minimum optimistic value over all
fractional matchings, it is enough to show that
\[
    \min_{j \in [k]} \Vstar{i}{\widehat y^j}
    \;\le\;
    A_i(\theta) + f_i(\theta).
\]

Let
\[
    T_i(\theta)
    :=
    \bigl\{
        o \in M :
        \text{$o$ is liked by some agent in $L_i(\theta)$}
    \bigr\}.
\]
We first show that
\begin{equation}
\label{eq:T-bound}
    |T_i(\theta)| \;\le\; k\,A_i(\theta).
\end{equation}

Fix an item $o \in T_i(\theta)$, and choose an agent
$a \in L_i(\theta)$ such that $(a,o) \in E$.
Since $\widetilde x_a < \theta$ in the \emph{final} guiding process and guiding
loads are monotone over time, agent $a$ has current load $< \theta \le \gamma$
throughout the entire execution.
Hence, whenever item $o$ is processed, class $i$ is active during the
equal-filling phase of EFTT.
Because EFTT splits the item equally among all active classes during this
phase, class $i$ receives at least a $1/k$ fraction of item $o$.

Moreover, within class $i$, EFTT always water-fills among the minimum-load
neighbors of the arriving item.
Since agent $a$ likes $o$ and has current load $< \theta$ at all times, every
infinitesimal piece of $o$ allocated to class $i$ is assigned to an agent of
current load at most that of $a$, and hence of current load $< \theta$.
Therefore the entire fraction of $o$ given to class $i$ contributes to the
quantity $A_i(\theta)$.
Summing over all items in $T_i(\theta)$ gives \eqref{eq:T-bound}.

Now consider the total optimistic value
$\sum_{j=1}^k \Vstar{i}{\widehat y^j}$.
For each $j$, in an optimal fractional matching realizing
$\Vstar{i}{\widehat y^j}$, the contribution of agents in $L_i(\theta)$ can only
come from items in $T_i(\theta)$, because these are exactly the items liked by
at least one agent in $L_i(\theta)$.
Hence the total contribution of all agents in $L_i(\theta)$ across all $k$
bundles is at most
\[
    \sum_{o \in T_i(\theta)} \sum_{j=1}^k \widehat y^j_o
    \;\le\;
    |T_i(\theta)|
    \;\le\;
    k\,A_i(\theta).
\]
On the other hand, there are exactly $f_i(\theta)$ agents in $H_i(\theta)$, and
each such agent can contribute at most $1$ to each of the $k$ optimistic
matchings. Thus the total contribution of agents in $H_i(\theta)$ across all
$k$ bundles is at most
\[
    k\,f_i(\theta).
\]
Combining the low-load and high-load contributions,
\[
    \sum_{j=1}^k \Vstar{i}{\widehat y^j}
    \;\le\;
    k\,A_i(\theta) + k\,f_i(\theta).
\]
Therefore,
\[
    \min_{j \in [k]} \Vstar{i}{\widehat y^j}
    \;\le\;
    \frac{1}{k}\sum_{j=1}^k \Vstar{i}{\widehat y^j}
    \;\le\;
    A_i(\theta) + f_i(\theta).
\]
Since $Z$ was arbitrary, this proves \eqref{eq:prop-envelope}.
\hfill$\diamond$

\medskip

We now integrate the envelope \eqref{eq:prop-envelope} against a suitable
kernel.

Define
\[
    F_\gamma(\theta)
    :=
    e^\theta \int_\theta^\gamma e^{-z} p'(z)\,dz,
    \qquad
    c_\gamma(\theta)
    :=
    -F_\gamma'(\theta)
    =
    p'(\theta) - F_\gamma(\theta)
    \qquad (0 \le \theta \le \gamma).
\]
A direct calculation shows that $p$ is concave on $[0,\infty)$; equivalently,
$p'$ is nonincreasing on $[0,\infty)$.
Therefore, for every $\theta \in [0,\gamma]$,
\[
    F_\gamma(\theta)
    \;\le\;
    p'(\theta)\,e^\theta \int_\theta^\gamma e^{-z}\,dz
    =
    p'(\theta)\bigl(1-e^{\theta-\gamma}\bigr)
    \;\le\;
    p'(\theta),
\]
and hence
\[
    c_\gamma(\theta) \;\ge\; 0.
\]
Also, since $c_\gamma = -F_\gamma'$, we have
\[
    \int_\theta^\gamma c_\gamma(u)\,du
    =
    F_\gamma(\theta) - F_\gamma(\gamma)
    =
    F_\gamma(\theta),
\]
because $F_\gamma(\gamma)=0$.
Consequently,
\begin{equation}
\label{eq:kernel-identity}
    c_\gamma(\theta) + \int_\theta^\gamma c_\gamma(u)\,du
    =
    p'(\theta)
    \qquad
    (0 \le \theta \le \gamma).
\end{equation}

Multiply \eqref{eq:prop-envelope} by $c_\gamma(\theta)$ and integrate over
$\theta \in [0,\gamma]$:
\[
    \prop_i \int_0^\gamma c_\gamma(\theta)\,d\theta
    \;\le\;
    \int_0^\gamma c_\gamma(\theta) A_i(\theta)\,d\theta
    +
    \int_0^\gamma c_\gamma(\theta) f_i(\theta)\,d\theta.
\]
By Fubini's theorem and \eqref{eq:kernel-identity},
\[
\begin{aligned}
    \int_0^\gamma c_\gamma(\theta) A_i(\theta)\,d\theta
    &=
    \int_0^\gamma c_\gamma(\theta)
    \left(\int_0^\theta f_i(z)\,dz\right)\,d\theta \\
    &=
    \int_0^\gamma
    \left(\int_z^\gamma c_\gamma(\theta)\,d\theta\right) f_i(z)\,dz.
\end{aligned}
\]
Hence
\[
\begin{aligned}
    \prop_i \int_0^\gamma c_\gamma(\theta)\,d\theta
    &\le
    \int_0^\gamma
    \left(
        \int_z^\gamma c_\gamma(\theta)\,d\theta
        +
        c_\gamma(z)
    \right)
    f_i(z)\,dz \\
    &=
    \int_0^\gamma p'(z)\,f_i(z)\,dz.
\end{aligned}
\]
Finally,
\[
    \int_0^\gamma c_\gamma(\theta)\,d\theta
    =
    F_\gamma(0)
    =
    \int_0^\gamma e^{-z}p'(z)\,dz
    =
    \rho(\gamma).
\]
Therefore,
\[
    \rho(\gamma)\,\prop_i
    \;\le\;
    \int_0^\gamma p'(z)\,f_i(z)\,dz
    \;\le\;
    \int_0^\infty p'(z)\,f_i(z)\,dz
    \;\le\;
    \EE[V_i(X)],
\]
where the last inequality is \eqref{eq:ocs-lower-bound}.

Since the class $i$ was arbitrary, the above holds for every class
$i \in [k]$, proving that the algorithm is $\rho(\gamma)$-CPROP in
expectation.
\end{proof}

\end{document}